\newif\ifarxiv
\arxivtrue

\newif\ifllncs
\llncsfalse 

\ifarxiv
\documentclass[a4paper,british]{article}
\usepackage[top=2.7cm,margin=1in]{geometry}
\else

\ifllncs
\documentclass[a4paper,UKenglish,cleveref, autoref, thm-restate, runningheads,envcountsame]{llncs} 
\usepackage[T1]{fontenc}
\else
\documentclass[a4paper,UKenglish,cleveref, autoref, thm-restate]{lipics-v2021} 
\fi

\fi

\usepackage{babel}
\usepackage[utf8]{inputenc}

\ifarxiv
\fi

\ifarxiv 
\usepackage[hyphens]{url}
\usepackage[hidelinks]{hyperref}
\fi
\usepackage{subcaption,comment} 
\usepackage{tikz}

\usepackage{mathtools}

\usepackage{graphicx}

\ifllncs

\makeatletter
\if@envcntsame
\let\c@conjecture\relax
\newcounter{conjecture}
\fi
\makeatother

\usepackage[hidelinks]{hyperref}
\usepackage{color}

\fi

\usepackage{amsmath,amsfonts,amssymb,amsthm}
\usepackage{thm-restate}

\usepackage{microtype,xspace,wrapfig,multicol} 
\usepackage[textsize=tiny,color=lightgray]{todonotes} 
\usepackage[normalem]{ulem} 
\usepackage{tabularx}
\makeatletter
\newcommand{\specificthanks}[1]{\@fnsymbol{#1}}
\makeatother

\newcommand{\BB}{\text{BB}}
\newcommand{\beatit}{15\xspace}

\makeatletter
\newcommand\footnoteref[1]{\protected@xdef\@thefnmark{\ref{#1}}\@footnotemark}
\makeatother

\title{Hardness of busy beaver  value $\BB(\beatit)$}
\ifarxiv
\else

\ifllncs

\else
\ccsdesc[100]{Theory of computation~Models of computation}
\ccsdesc[100]{Mathematics of computing~Discrete mathematics}
\author{ }{ }{ }{ }{ }
\keywords{Busy beaver, Turing machines, Erd\H{o}s' conjecture, mathematical hardness, simulation}
\Copyright{ } 
\authorrunning{ }
\fi

\fi

\ifarxiv
\author{Tristan St\'{e}rin\textsuperscript{\specificthanks{1}} \hspace{2em} Damien Woods\thanks{Hamilton Institute and Department of Computer Science, Maynooth University, Ireland. Research supported by European Research Council (ERC, grant agreement No 772766, Active-DNA project) under the European Union's Horizon 2020 research and innovation programme, European Innovation Council  (EIC, DISCO, No 101115422); and Science Foundation Ireland (SFI) under Grant number 18/ERCS/5746.} }
\fi
\date{}

\ifarxiv
\newtheorem{theorem}{Theorem}[]
\newtheorem{definition}[theorem]{Definition}
\newtheorem{lemma}[theorem]{Lemma}

\newtheorem{conjecture}{Conjecture} 

\newtheorem{example}[theorem]{Example}

\fi

\usepackage{tikz}

\newcommand{\N}{\mathbb{N}}
\newcommand{\Z}{\mathbb{Z}}

\newcommand{\mulTwo}{{\bf\textrm{mul2}}\xspace}
\newcommand{\FSTF}{\ensuremath{F}}
\newcommand{\FSTG}{\ensuremath{G}}
\newcommand{\mulTwoF}{{\bf mul2\_F}\xspace}
\newcommand{\mulTwoG}{{\bf mul2\_G}\xspace}
\newcommand{\checkTwo}{{\bf find\_2}\xspace}
\newcommand{\checkHalt}{{\bf check\_halt}\xspace}
\newcommand{\rewind}{{\bf rewind}\xspace}

\newcommand{\mulTwoFS}{{\bf mul2\_F\_}}
\newcommand{\mulTwoGS}{{\bf mul2\_G\_}}
\newcommand{\checkTwoS}{{\bf find\_2\_}}
\newcommand{\checkHaltS}{{\bf check\_halt\_}}
\newcommand{\rewindS}{{\bf rewind\_}}

\newcommand{\Eds}{Erd\H{o}s\xspace}
\newcommand{\Mf}{\ensuremath{M_{5,4}}\xspace}
\newcommand{\Mt}{\ensuremath{M_{15,2}}\xspace}

\newcommand\bolu[1]{{\pmb{\underline{#1}}}}

\begin{document}

\maketitle
\begin{abstract}
  The busy beaver value BB($n$) is the maximum number of steps made by any $n$-state, 2-symbol deterministic halting  Turing machine starting on blank tape. 
  The busy beaver function $n \mapsto \BB(n)$ is uncomputable and, from below, only 4 of its values, $\BB(1)\, \dots\, \BB(4)$, are known to date.
  This leads one to ask: from above, what is the smallest BB value that encodes a major mathematical challenge?
  Knowing BB(4,888) has been shown by Yedidia and Aaronson \cite{Yedidia2016} to be at least as hard as solving Goldbach's conjecture, with a subsequent improvement, as yet unpublished, to  BB(27)~\cite{Charles2013,BusyBeaverFrontier}.
  We prove that knowing BB(15) is at least as hard as solving the following Collatz-related conjecture by Erd\H{o}s, open since 1979~\cite{ErdosPowers2}: for all $n>8$ there is at least one digit~2 in the base 3 representation of $2^n$.
  We do so by constructing an explicit 15-state, 2-symbol Turing machine that halts if and only if the conjecture is false. This 2-symbol Turing machine simulates a conceptually simpler 5-state, 4-symbol machine which we construct first.
  This makes, to date, BB(15) the smallest busy beaver value that is related to a \textit{natural} open problem in mathematics,  bringing to light one of the many challenges underlying the quest of knowing busy beaver values.
\end{abstract}
%

\section{Introduction}

In the theory of computation, there lies a strange and complicated beast, the busy beaver function. This function tracks a certain notion of algorithmic complexity, namely, the maximum number of steps any halting algorithm of a given program-size may take. Although formulated in terms of Turing machines, the underlying notion of ``maximum algorithmic bang for your buck''  that it captures could be defined in any reasonable programming language.

The busy beaver function $n \mapsto \BB(n)$
was introduced by Tibor Rad\'{o} in 1962 and corresponds to the maximum number of steps made by a halting deterministic Turing machine with $n$ states and 2 symbols starting from blank input~\cite{Rado_1962}.\footnote{Rad\'{o} originally used the notation $S(n)$. The more modern~\cite{BusyBeaverFrontier} notation $\mathrm{BB}(n)$ can be a source of confusion since some authors, for example \cite{HARLAND2022368,Pavel_discorvery}, use BB to mean Rad\'{o}'s $\Sigma$ which counts the number of 1s on the final tape of a halting Turing machine \cite{Rado_1962}.}
It was generalised by Brady~\cite{10.5555/57249.57258,michel2019busy} 
to machines with $k$ symbols; $\BB(n,k)$. Busy beaver functions, i.e. $n,k \mapsto \BB(n,k)$, or $n \mapsto \BB(n, k)$ for fixed $k \geq 2$, are not computable.
Otherwise the halting problem on blank tape would be computable: take any machine with $n$ states and $k$ symbols, run it for $\BB(n,k)+1$ steps; if it has not halted yet, we know that it will never halt. They also dominate any computable function.

To date, only four non-trivial\footnote{Trivial values include those where $n=1$ since for all $k \geq 1$, $\BB(1,k) = 1$.} busy beaver values are known: $\BB(2) = 6$, $\BB(3) = 21$, $\BB(4) = 107$ and $\BB(2,3) = 38$ \cite{michel2019busy,PMichel_website}. Machines that halt without input after so many steps that they beat previously known record-holders, with the same $n$ and $k$, are called busy beaver \textit{champions} and they give lower bounds for $\BB(n)$ or  $\BB(n,k)$. It is conjectured that $\BB(5) = \text{47,176,870}$ as there is an explicit 5-state 2-symbol champion that halts after $\text{47,176,870}$ steps\footnote{%
  \scriptsize{
    \begin{tabular}{l|l|l|l|l|l|}
        & \textbf{A} (init)  & \textbf{B}         & \textbf{C}         & \textbf{D}         & \textbf{E}         \\ \hline
      0 
        & 1\, R\, \textbf{B} & 1\, R\, \textbf{C} & 1\, R\, \textbf{D} & 1\, L\, \textbf{A} & Halt               \\ \hline
      1 & 1\, L\, \textbf{C} & 1\, R\, \textbf{B} & 0\, L\, \textbf{E} & 1\, L\, \textbf{D} & 0\, L\, \textbf{A} \\ \hline
    \end{tabular}}\\ 
  Current busy beaver champion~\cite{Marxen_1998} for machines with $5$ states (\textbf{A}--\textbf{E}) and $2$ symbols (0,1), it halts in $\text{47,176,870}$ steps starting from all-$0$ input and state $\textbf{A}$, which gives the lowerbound $\BB(5) \geq \text{47,176,870}$. Proving that bound tight is the goal of the bbchallenge project \cite{bbchallengeBusyBeaver}.
} \cite{Marxen_1998,BusyBeaverFrontier}.
It is known \cite{Pavel_discorvery,PMichel_website} that $
  \BB(6) \geq 10 \uparrow \uparrow 15 = 10^{{\scriptstyle 10^{\cdot^{\cdot^{\cdot^{\scriptstyle 10}}}}}}$, a tower of 15 powers of 10, well beyond $10^{80}$, the estimated
number of atoms in the observable universe. Also, $\BB(2,6) \geq 10 \uparrow \uparrow ( 10 \uparrow \uparrow 10^{10^{315}})$ \cite{Pavel_discorvery2,PMichel_website} and $\BB(3,4) \geq \text{Ack}(14)$, with the Ackerman number defined as $\text{Ack}(n) = n \uparrow^n n$  \cite{Pavel_discorvery3,Ack14BlogPost}.
These results highlight the sheer growth of the busy beaver function (it provably grows faster than any computable function \cite{Rado_1962}). However, lower bounds like this one do not give us a concrete sense of how hard it is to know busy beaver values, we just know that they quickly become ``huge''.


Can we give a more formal notion of {\em hardness} of finding busy beaver values?
Indeed, for specific values of $n$ and $k$, how hard is it to find $\BB(n,k)$?
A recent trend, that we follow here, seeks to formally relate values of $\BB(n,k)$ to notoriously hard mathematical problems.
For instance, one can imagine~\cite{chaitin1987computing} designing a Turing machine that, starting from blank tape, will halt if and only if it finds a counterexample to Goldbach's conjecture (every even integer greater than 2 is the sum of two primes).
Such a machine was actually built, using 4,888 states and $2$ symbols~\cite{Yedidia2016}. This result implies that knowing the value of $\BB(\textrm{4,888})$ would allow us to computably decide Goldbach's conjecture: run the machine for $\BB(\textrm{4,888})+1$ steps -- which, as we've seen, must be an unbelievably huge number -- and if it has halted before, then the conjecture is false otherwise it is true. The result was later claimed to be improved to a 27-state 2-symbol machine \cite{Charles2013,BusyBeaverFrontier}, which, subject to being proved correct, would be the smallest busy beaver value, before our work, that relates to a \textit{natural} mathematical problem. Similar efforts led to the construction of a 5,372-state 2-symbol machine that halts if and only if the Riemann hypothesis is false~\cite{Yedidia2016}; with a later claimed improvement to $744$ states~\cite{BusyBeaverFrontier}, hence knowing the value of $\BB(744)$ is at least as hard as solving the Riemann hypothesis.


\subsection{Results and discussion}

Here, we continue the approach of relating small busy beaver values to hard mathematical problems  towards obtaining insight into the location of the  frontier between knowable and unknowable busy beaver values.
Our particular approach gives upper bounds on the smallest counterexample to such problems.
We do this by giving machines that search for counterexamples to an open, Collatz-related conjecture formulated in 1979 by Erd\H{o}s:
\begin{conjecture}[Erd\H{o}s~\cite{ErdosPowers2}]\label{conj:Erdos}
  \normalfont
  For all natural numbers $n > 8$ there is at least one digit $2$ in the base 3 representation of~$2^n$.
\end{conjecture}

\noindent In Section~\ref{sec:erdos} we discuss existing literature on the conjecture as well as its relationship with the Collatz and weak Collatz conjectures \cite{lagariasErdos,DUPUY2016268,dimitrov2021powers,TerenceTaoBlog}.

The main technical contribution of this paper is to prove the following theorem:
\begin{restatable}[]{theorem}{thmtwosymbols}
  \label{thm:two}
  \normalfont
  There is an explicit $\beatit$-state $2$-symbol Turing machine that
  halts if and only if Erd\H{o}s' conjecture is false.
\end{restatable}
\noindent That Turing machine, called $\Mt$,  is given in Figure~\ref{fig:two} and proven correct in Section~\ref{sec:two}.
The proof shows that $\Mt$ {\em simulates}, in a tight (linear time) fashion, an intuitively simpler machine with 5 states and 4 symbols, that we call $\Mf$
(Figure~\ref{fig:four})
and whose behaviour is proven correct in Section~\ref{sec:four}:

\begin{restatable}[]{theorem}{thmfoursymbols}
  \label{thm:four}
  \normalfont
  There is an explicit $5$-state $4$-symbol Turing machine that
  halts if and only if Erd\H{o}s' conjecture is false.
\end{restatable}

From these theorems, we get that knowing the value of $\BB(\beatit)$ is at least as hard as solving Erd\H{o}s' conjecture since $\BB(\beatit)$ gives a finite---although particularly impractical---algorithmic procedure to decide the conjecture as it upperbounds the set of values that need to be considered:

\begin{restatable}[]{corollary}{ErdosIsFinite}
  \label{cor:ErdosIsFinite}
  \normalfont
  Erd\H{o}s' conjecture is equivalent to the following conjecture over a finite set: for all $8 < n \leq \min (\BB(\beatit), \BB(5,4)  )$ there is at least one digit $2$ in the base $3$ representation of~$2^n$.
\end{restatable}

Finally, we also get an upper bound on the smallest counterexample to Erd\H{o}s conjecture, if it exists:

\begin{restatable}[]{corollary}{BBlowerBound}
  \label{cor:BBlowerBound}
  \normalfont
  Let $x\in\N$ be the smallest counterexample to Erd\H{o}s conjecture, if it exists. Then we have: $\BB(15) \geq \log_2 x$ and $\BB(5,4) \geq \log_2 x$.
\end{restatable}

\ifllncs
\else
  \subparagraph*{Discussion.}
\fi
Why are these theorems important? They give us more knowledge on the busy beaver function: we get that values as small as BB(15) are related to hard open problems in mathematics and thus correspondingly hard to know. Our results make $\BB(\beatit)$ the smallest busy beaver value linked to an open problem in mathematics, a significant improvement on the $\BB(\text{4,888})$ and $\BB(\text{5,372})$ results cited earlier (that have a proof of correctness), and the
$\BB(27)$ and $\BB(744)$ results (that as yet lack a published proof of correctness).

An even more drastic way to establish hardness is to find an $n$ for which $\BB(n)$ is independent of some standard set of axioms such as PA, or ZFC. In that spirit, a 7,910-state machine whose halting problem is independent of ZFC was constructed \cite{Yedidia2016} and this was later improved to 748 states~\cite{BB748Thesis,BusyBeaverFrontier}. The $748$-state machine explicitly looks for a contradiction in ZFC (such as a proof of $0=1$) which is, by G\"{o}del's second incompleteness theorem, independent of ZFC. Aaronson~\cite{BusyBeaverFrontier} conjectures that $\BB(10)$ is independent of PA and $\BB(20)$ is independent of ZFC meaning that the frontier between knowable and unknowable busy beaver values could be as low as $\BB(10)$ if we limit ourselves to typical inductive proofs. With our work, any independence result on Erd\H{o}s' conjecture would automatically transfer to busy beaver value BB(15) and higher.
As for Erd\H{o}s' conjecture, Corollary~\ref{cor:ErdosIsFinite} gives a finite bound on the number of counter-examples needed to check in order to prove or refute the conjecture. However, this is not a practical bound: first, we don't know its value, and second, it must be so astronomical as to not give a tractable procedure to prove or disprove the conjecture.
Moreover, since we now know that BB(15) {\em embeds} Erd\H{o}s' conjecture, and that it could embed other horrible beasts, finding BB(15) itself could indeed be much harder than answering the conjecture.
However, it is not unthinkable that methods for automatically deciding
whether certain Turing machines halt or not \cite{bbchallengeBusyBeaver,Marxen_1998,Skelet_website} could one day solve the halting problem of the machines constructed in this article (or future improved machines with less states) and thus solve Erd\H{o}s' conjecture that way.
This approach  might highlight unforeseen links between theoretical computer science and number theory.
As we've noted, some claimed results on the busy beaver function come with proofs, but some do not~\cite{BusyBeaverFrontier}. We advocate for proofs, although we acknowledge the challenges in providing human-readable correctness proofs for small programs that are, or almost are, program-size optimal.\footnote{The situation can be likened to the hunt for small, and fast, universal Turing machines~\cite{WoodsNearySurvey}, or simple models like Post tag systems, with earlier literature often missing proofs of correctness, but some later papers using induction on machine configurations to do the job, for example refs~\cite{neary2009four,WoodsNeary2006B}.}
Here, the proof technique for correctness of our $\BB(\beatit)$ candidate $\Mt$ amounts to proving by induction that $\Mt$ \textit{simulates} (in a tight  sense via Definition~\ref{def:sim} and Lemma~\ref{lem:sim}) another Turing machine $\Mf$ (giving Theorem~\ref{thm:two}).
$\Mf$ exploits the existence of a tiny finite state transducer (FST), in Figure~\ref{fig:fst}, for multiplication by~2 in base~3. We directly prove  $\Mf$'s correctness by induction on its time steps (giving Theorem~\ref{thm:four}).
A Turing machine simulator, \texttt{bbsim}, was built in order to test our constructions\footnote{\label{ft:onlineSimulator}Simulator and machines available here: \url{https://github.com/tcosmo/bbsim}}. The reader is invited to run the machines of this paper through the simulator.

\subsection{Erd\H{o}s' conjecture and its relationship to the Collatz and weak Collatz conjectures}\label{sec:erdos}

\begin{figure}[t]
  \center
  \ifllncs
\else
    \includegraphics[width=0.9\textwidth]{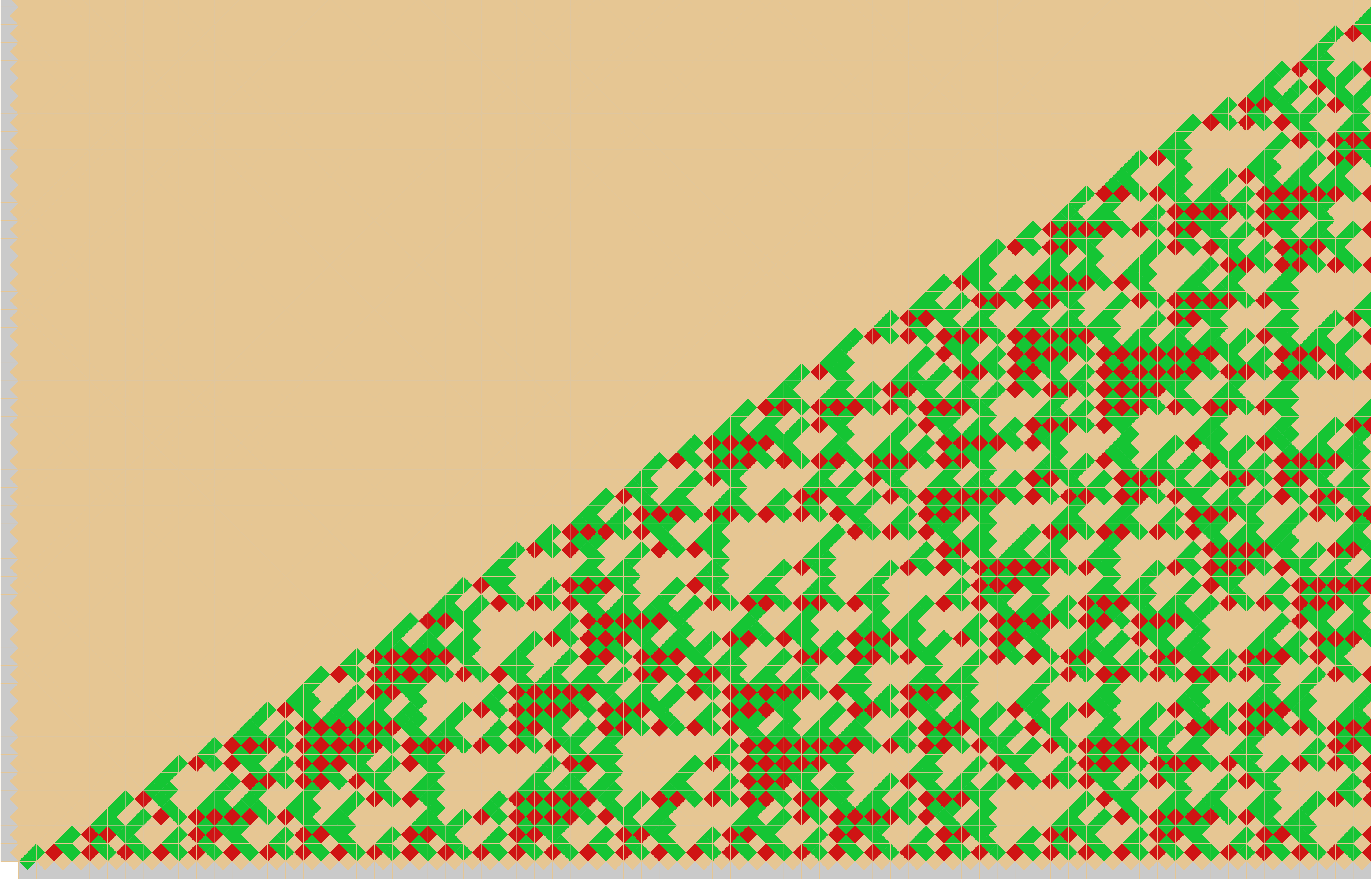}
  \fi
  \caption{The first 75 powers of two assembled in base 3 by the size-6 Wang tile set introduced in~\cite{thesis_sixtiles,mawatam}.
    Reading left-to-right, each column of glues (colours) corresponds to a power of two:
    beige glues represent ternary digit $0$, green glues $1$ and red glues $2$.
    For instance, the rightmost column encodes (from top to bottom) $110210021020202011202012000020112001021021222022_3 = 2^{75}$. The complexity of the patterns illustrates the complexity of answering Erd\H{o}s conjecture which amounts to asking if each glue-column (except for the few first) has at least one red glue.}\label{fig:erdos}
\end{figure}

A first obvious fact about Conjecture~\ref{conj:Erdos} is that the bound $n > 8$ is set so as to exclude the three known special cases: $1$, $4$ and $256$ whose ternary representations are respectively $1$, $11$ and $100111$.
Secondly, since having no digit $2$ in ternary is equivalent to being a sum of distinct powers of three, Conjecture~\ref{conj:Erdos} can be restated as: for all $n>8$, the number $2^n$ is not a sum of distinct powers of~$3$.

The conjecture has been studied by several authors. Notably, Lagarias~\cite{lagariasErdos} showed a result stating that, in some sense, the set of powers of $2$ that omits the digit $2$ in base three, is small. In \cite{DUPUY2016268}, the authors showed that, for $p$ and $q$ distinct primes, the digits of the base $q$ expansions of $p^n$ are equidistributed on average (averaging over $n$) which in our case suggests that digits $0,1,2$ should appear in equal proportion in the base 3 representation of $2^n$. Dimitrov and Howe \cite{dimitrov2021powers} showed that $1$, $4$ and $256$ are the only powers of $2$ that can be written as the sum of at most twenty-five distinct powers of $3$.

The Collatz conjecture states that iterating the Collatz function $T$ on any $x\in\mathbb{N}$ eventually yields 1, where $T(x) = x/2$ if $x$ is even and $T(x) = (3x+1)/2$ if $x$ is odd. The weak Collatz conjecture states that if $T^k(x) = x$ for some $k\geq1$ and $x$ a natural number then $x\in\{0,1,2\}$.

Although solving the weak Collatz conjecture, given current knowledge, would not directly solve Erd\H{o}s' conjecture,
intuitively, Erd\H{o}s' conjecture seems to be the simpler problem of the two. Indeed, Tao~\cite{TerenceTaoBlog} justifies calling Erd\H{o}s' conjecture a ``toy model'' problem for the weak Collatz conjecture by giving the number-theoretical reformulation that there are no integer solutions to $2^n = 3^{a_1} + 3^{a_2} + \dots + 3^{a_k}$ with $n > 8$ and $0 \leq a_1 < \dots < a_k$, which in turn seems like a simplification of a statement equivalent to the weak Collatz conjecture, also given in \cite{TerenceTaoBlog} (Conjecture 3; Reformulated weak Collatz conjecture).

The three conjectures have been encoded using the same size-6 Wang tile set by Cook, Stérin and Woods~\cite{thesis_sixtiles,mawatam}.
As Figure~\ref{fig:erdos} shows, their construction illustrates the complexity of the patterns occurring in ternary representations of powers of 2 which gives a sense of the complexity underlying Erd\H{o}s' conjecture. Beyond making complex, albeit pretty, pictures there is deeper connection here: the small tile set of~\cite{thesis_sixtiles,mawatam} can be shown to simulate the base 3, \mulTwo, Finite State Transducer that we introduce in Section~\ref{sec:four}, Figure~\ref{fig:fst}, and our small Turing machines use it to look for counterexamples to Erd\H{o}s' conjecture. The tile set also simulates  the \textit{inverse} of the \textbf{mul2} FST {(which computes the operation $x\mapsto x/2$ in ternary)} which was used to build a 3-state 4-symbol non-halting machine that runs the Collatz map on any ternary input \cite{michel2014simulation} and finally, {it also simulates} the \textit{dual} of that FST (which computes the operation $x\mapsto 3x + 1$ in binary) which can be used to simultaneously compute the Collatz map both in binary and ternary~\cite{Collatz2}. So these three closely-related FSTs are all encoded within that small tile set, that in turn encodes the three conjectures.

\subsection{Future work}
This paper opens a number of avenues for future work.

We've given a finite bound on the first counterexample to Erd\H{o}s' conjecture, an obvious future line of work is to improve that cosmologically large bound by shrinking the program-size of our 15-state machine. An avenue for reducing the number of states needed to encode Erd\H{o}s' conjecture is to study small Turing machines and look for machines with the expected behaviour ``in the wild''\footnote{Such an effort has been done for 5-state Turing machines \cite{bbchallengeBusyBeaver}.}.

It would be interesting to design small Turing machines that look for counterexamples to the weak Collatz conjecture. That way we could relate the fact of knowing busy beaver values to solving that  notoriously hard conjecture. We have already found 124-state 2-symbol, and 43-state 4-symbol, machines that look for such counterexamples\footnote{These machines are also available here: \url{https://github.com/tcosmo/bbsim}} but we would like to further optimise them (i.e. reduce their number of states) before formally proving that their behaviour is correct.

There are certain other open problems amenable to BB-type encodings.
For instance, the Erd\H{o}s–Mollin–Walsh conjecture states that there are no three consecutive powerful numbers, where $m\in\mathbb{N}$ is {\em powerful} if $m = a^2 b^3$ for some $a,b\in\mathbb{N}$. It would not be difficult to give a small Turing Machine, that enumerates unary encoded natural numbers, i.e.\ candidate $m$ and $a,b$ values, and runs the logic to check for counterexamples. Algorithms that bounce back and forth on a tape marking off unary strings have facilitated the finding of incredibly small universal Turing machines~\cite{neary2009four,rogozhin1996small,WoodsNearySurvey}, so stands a chance to work well here too, although we wouldn't expect it to beat our main results in terms of program-size.

In fact, any problem that can be expressed as a $\Pi_1$ sentence~\cite{BusyBeaverFrontier} should be relatable to a value of $\BB$ in the same way that we did for Erd\H{o}s' conjecture, although many problems would presumably yield unsuitably large program-size.
It is also worth noting that some problems seem out of reach of the busy beaver framework: for instance the Collatz conjecture, as to date, there are no known ways for an algorithm to recognise whether or not an arbitrary trajectory is divergent, making it seemingly impossible for a program to search for a divergent counterexample\footnote{Said otherwise: given current knowledge there are no reason to believe that the set of counterexamples to the Collatz conjecture is recursively enumerable.}. Aaronson~\cite{BusyBeaverFrontier} gives a number of open problems and conjectures on the structure of BB functions.

\section{Definitions: busy beaver Turing machines}

Let $\mathbb{N} = \{0,1,2\dots\}$ and $\mathbb{Z}=\{\dots, -1, 0, 1, \dots\}$.
We consider Turing machines that  are deterministic, have a single bi-infinite tape with tape cells indexed by $\mathbb{Z}$,
and a finite alphabet of $k$ tape symbols that includes a special `blank' symbol (we use $\#$ or $b$). For readability, we use bold programming-style names, for example \checkHalt.
We use `(init)' to denote the initial state.  Each state has $k$ transitions, one for each of the $k$ tape symbols that might be read by the tape head, and a transition is either (1) the halting instruction `Halt' or (2) a triple: write symbol,  tape head move direction (L or R), and  next state. In the busy beaver setting used throughout this paper, we start machines in their initial state on an all-blank tape and, at each time step, according to what symbol is read on the tape, the specified transition is performed, until `Halt' is encountered (if ever).

Let $\text{TM}(n,k)$ be the set of such $n$-state, $k$-symbol Turing machines. Given a machine $M \in \text{TM}(n,k)$, let $s(M)$ be the number of transitions it executes before halting, including the final Halt instruction\footnote{As in \cite{BusyBeaverFrontier}, we took the liberty of not defining a symbol to write and a direction to move the tape head to when the halting instruction is performed as this does not change the number of transitions that the machine executed.} and let $s(M) = \infty$ if $M$ does not halt. Then, $\text{BB}(n,k)$ is defined~\cite{BusyBeaverFrontier} by
\[ \text{BB}(n,k) = \max_{M \in \text{TM}(n,k), \; s(M) < \infty} s(M) \]
In other words, $\text{BB}(n,k)$ is the number of steps by a machine in $\text{TM}(n,k)$ that runs the longest without halting. By convention,
$\text{BB}(n) = \text{BB}(n,2)$, this being the most classic and well-studied busy beaver function.

Some conventions: A busy beaver candidate is a Turing machine for which we don't currently {\em know} whether it halts or not on blank input. A busy beaver contender is a Turing machine that halts on blank input.
A busy beaver champion is a Turing machine that halts on blank input in more steps than any other {\em known} machine with the same number of states and symbols.

A Turing machine configuration is given by: the current state, the tape contents, and an integer tape head position.
We sometimes write configurations in the following condensed format: $\textbf{state\_name},\; \dots ***\bolu{*}*** \dots$ with $*$ any tape symbol of the machine and $\,\bolu{\;\,}\,$ for tape head position.
Let $c_1$ and $c_2$ two configurations of some machine $M$ and let $k\in\N$, we write $c_1 \vdash^k_M c_2$ to mean that $M$ transitions from $c_1$ to $c_2$ in $k$ steps. We write  $\vdash$ (without $M$) if $M$ is clear from the context.


\section{Five state, four symbol Turing machine}\label{sec:four}
In this section we prove Theorem~\ref{thm:four}.
To do this we define, and prove correct, a $5$-state $4$-symbol Turing machine that searches for a counterexample to Erd\H{o}s' conjecture, the machine is given in Figure~\ref{fig:four}.
The construction begins with the \mulTwo finite state transducer (FST) in Figure~\ref{fig:fst}.

\begin{figure}[!htb]
\begin{center}
\ifllncs
\vspace{-3ex}
\begin{tikzpicture}[scale=0.15]
\else
\begin{tikzpicture}[scale=0.2]
    \fi
    \tikzstyle{every node}+=[inner sep=0pt]
    \draw [black] (19.1,-22.8) circle (3);
    \draw (19.1,-22.8) node {$\FSTF$};
    \draw [black] (37.1,-22.8) circle (3);
    \draw (37.1,-22.8) node {$\FSTG$};
    \draw [black] (21.585,-21.13) arc (117.76503:62.23497:13.985);
    \fill [black] (34.62,-21.13) -- (34.14,-20.31) -- (33.67,-21.2);
    \draw (28.1,-19.02) node [above] {$2:1$};
    \draw [black] (17.777,-20.12) arc (234:-54:2.25);
    \draw (21.8,-16) node [above] {$0:0$};  
    \fill [black] (20.42,-20.12) -- (21.3,-19.77) -- (20.49,-19.18);
    \draw [black] (20.423,-25.48) arc (54:-234:2.25);
    \draw (21.8,-29.8) node [below] {$1:2$};   
    \fill [black] (17.78,-25.48) -- (16.9,-25.83) -- (17.71,-26.42);
    \draw [black] (35.777,-20.12) arc (234:-54:2.25);
    \draw (39.8,-16) node [above] {$2:2$};  
    \fill [black] (38.42,-20.12) -- (39.3,-19.77) -- (38.49,-19.18);
    \draw [black] (34.41,-24.119) arc (-68.79964:-111.20036:17.448);
    \fill [black] (21.79,-24.12) -- (22.36,-24.87) -- (22.72,-23.94);
    \draw (28.1,-25.8) node [below] {$0:1$};
    \draw [black] (38.423,-25.48) arc (54:-234:2.25);
    \draw (39.8,-29.8) node [below] {$1:0$};   
    \fill [black] (35.78,-25.48) -- (34.9,-25.83) -- (35.71,-26.42);
    \draw [black] (10.8,-22.8) -- (16.1,-22.8);
    \fill [black] (16.1,-22.8) -- (15.3,-22.3) -- (15.3,-23.3);
\end{tikzpicture}
\end{center}
\caption{The \mulTwo Finite State Transducer that multiplies a reverse-ternary represented number (base-3 written in reverse digit order) by $2$. For example, the base 10 number $64_{10}$  is $2101_3$ in base~3, which we represent in reverse-ternary with a leading zero to give the input $10120$, which in turn yields the FST output $20211$; the reverse-ternary of~$128_{10}$. Transition arrows are labelled $r:w$ where $r$ is the read symbol and $w$ is the write symbol.} \label{fig:fst}
\end{figure}

A similar FST, and its `dual', can be used to compute iterations of the Collatz map~\cite{Collatz2} (Appendix~B).
The fact that there is an FST that multiplies by 2 in base 3 is not surprising, since there is one for any affine transformation in any natural-number base $\geq 1$~\cite{baseChanges}.
However, in this section, we will exploit  \mulTwo's small size.
We begin by proving its behaviour is correct (by reverse-ternary we mean the base-3 representation written in reverse digit order):

\begin{lemma}\label{lem:fst}
    \normalfont
    Let $x\in\mathbb{N}$  and let $w = w_1 \dots w_n 0 \in \{0,1,2\}^{n}0$ be its reverse-ternary representation, with $n \geq 0$, and a single leading 0.
    Then, on input $w$ the \mulTwo FST outputs $\gamma = \gamma_1 \gamma_2 \dots \gamma_{n+1} \in \{0,1,2\}^{n+1}$ which represents  $y = 2x$ in reverse-ternary.
\end{lemma}

\begin{proof}
    We give an induction on ternary word length $n$, with the following induction hypothesis:
    given a word $w = w_1 \ldots w_n 0$ that represents $x\in\mathbb{N}$ (as in the lemma statement),
    if the FST reads from state $\FSTF$ then the operation $x \mapsto 2x$ is computed in reverse-ternary,
    and if the FST reads from state $\FSTG$ then  the operation $x \mapsto 2x+1$ is computed in reverse-ternary.

    For the base case, when $n=0$ we have $w=0$ and, when started from state $\FSTF$ the FST outputs $0$ and when started from state $\FSTG$ the FST outputs $1$ which corresponds to respectively applying $x \mapsto 2x$ and $x \mapsto 2x+1$.

    Let's assume that the induction hypothesis holds for $n$ and consider the base~3 word $w = w_1 \dots w_{n+1} 0 \in \{0,1,2 \}^{n+1}0$ that represents  some $x\in\mathbb{N}$.
    We first handle state $\FSTF$.
    There are three cases for the value of the least significant digit $w_1 \in \{0,1,2\}$.
    If $w_1=0$, from the FST state $\FSTF$ the first transition will output $0$ and return to state $\FSTF$.
    Then, from state $\FSTF$, by applying the induction hypothesis to the length-$m$ word $w_2 \dots w_{m+1}$, which represents the number $(x - w_1)/3=x/3$ in base $3$,
    we get that the FST output (on $w_2 \dots w_{n+1}$) is  the representation of $2\frac{x}{3}$.

    Then, by including the first output $0$,
    the complete output of the FST represents the number $0 + 3\cdot2\frac{x}{3} = 2x$ which is what we wanted.
    Similarly, if $w_1=1$ the FST outputs the representation of the number $2 + 3\cdot 2\frac{x-1}{3} = 2x$, or if $w_1=2$ the FST outputs the representation of $1 + 3 \cdot (2\frac{x-2}{3}+1) = 2x$ since it moves to state $\FSTG$ after the first transition. Likewise, if we start the FST in state $\FSTG$ the FST outputs:
    $1 + 3\cdot 2\frac{x}{3} = 2x+1$ if $w_0=0$, or
    $0 + 3\cdot(2\frac{x-1}{3}+1) = 2x+1$ if $w_0=1$, or
    $2 + 3\cdot(2\frac{x-2}{3} + 1) = 2x+1$ if $w_0 = 2$.
    In all the cases we get the result.
\end{proof}
\begin{figure}[!htb]
    \scriptsize\centering
    \setlength\tabcolsep{2pt}
    \begin{tabular}{l|l|l|l|l|l|}
                         & \textbf{\mulTwoF}   & \textbf{\mulTwoG}\hspace{.3ex} (init) & \textbf{\checkTwo}   & \textbf{\rewind}   & \textbf{\checkHalt} \\ \hline
        0                & 0 \, R \, \mulTwoF  & 1 \, R \, \mulTwoF                    & 0 \, L \, \checkTwo  & 0 \, L \, \rewind  & 1 \, R \, \rewind   \\
        1                & 2 \, R \, \mulTwoF  & 0 \, R \, \mulTwoG                    & 1 \, L \, \checkTwo  & 1 \, L \, \rewind  & 2 \, R \, \rewind   \\
        2                & 1 \, R \, \mulTwoG  & 2 \, R \, \mulTwoG                    & 2 \, L \, \rewind    & 2 \, L \, \rewind  & Halt                \\
        \#$\!\!$ (blank) & \#\, L \, \checkTwo & 1 \, R \, \mulTwoF                    & \#\, L \, \checkHalt & \#\, R \, \mulTwoF & 0 \, R \, \rewind   \\ \hline
    \end{tabular}
    \caption{5-state 4-symbol Turing machine  $\Mf$ that halts if and only if Erd\H{o}s' conjecture is false. The initial state of the machine is \textbf{\mulTwoG},  denoted `(init)'.
        The blank symbol is $\#$ and, since this is a busy-beaver candidate, the initial tape is empty:  $\ldots \# \# \underline{\#} \# \# \ldots$  (tape head underlined).
        Example~\ref{ex:four} shows the initial 333 steps of \Mf.
        States {\bf\mulTwoF} and {\bf \mulTwoG} implement states $\FSTF$ and $\FSTG$ of the ``$\mulTwo$'' FST in Figure~\ref{fig:fst}, that multiplies a reverse-ternary number by 2.
        The other states check whether the result is a counterexample to, or one of the three special cases of, Erd\H{o}s' conjecture.}\label{fig:four}
\end{figure}

Intuitively, the 5-state 4-symbol Turing machine $M_{5,4}$ in Figure~\ref{fig:four} works as follows.
Starting from all-\# tape and state $\textbf{\mulTwoG}$,
$M_{5,4}$  constructs successive natural number powers of $2$ in base~3 (in fact in reverse-ternary) by iterating the \mulTwo FST, which is embedded in its  states $\textbf{\mulTwoF}$ and $\textbf{\mulTwoG}$.
Then, for each power of two, $M_{5,4}$ checks that there is at least one digit $2$ using state $\textbf{\checkTwo}$.
If at least one digit $2$ is found, then, using state $\textbf{\rewind}$ the machine goes back to the start (left) and iterates on to the next power of $2$.
If no digit $2$ is found, such as in the three known special cases $1_{10}=1_3$, $4_{10} = 11_3$ and $256_{10}=100111_3$ (giving the condition $n>8$ in Erd\H{o}s' conjecture), then a counter is incremented on the tape (using state \textbf{\checkHalt}) and if this counter goes beyond value three the machine halts. The machine halts, iff we have found a counterexample to Erd\H{o}s' conjecture, a fact we prove formally below in Theorem~\ref{thm:four}.
\begin{example}
    \label{ex:four}
    Here, we highlight 11 out of the first 334 configurations of \Mf. Five of the first 16 configurations are shown on the left (read from top to bottom),
    and 6 out of configurations 17 to 334 are shown on the right.\footnoteref{ft:onlineSimulator}
    From step 5 onwards, the content of the tape is of the form $c\#w_1\dots w_n\#$ with $c,w_i \in\{0,1,2\}$, where $c$ represents a single-symbol counter keeping track of the $3$  special cases of Erd\H{o}s conjecture and $w_1\dots w_n$ is the reverse-ternary representation of a power of two. For instance, in the final configuration below we have 1101011202221 as the reverse-ternary representation of $2^{20} = 122021101011_3$:
\end{example}
\scriptsize
\begin{align*}
                  & \mulTwoG,   &       & \dots \#\#\#\bolu{\#}\#\#\dots\quad &  &  &  & \vdash^1     & \rewind,\;   & \quad\dots\#1\bolu{\#}11\#\dots               \\
    \vdash^5\quad & \rewind,    &       & \dots \#\# 0\bolu{\#}1\#\dots\quad  &  &  &  & \vdash^6     & \rewind,\;   & \quad\dots \#1\bolu{\#}22\#\dots              \\
    \vdash^4\quad & \rewind,    &       & \dots \# \#0\bolu{\#}2\#\dots\quad  &  &  &  & \vdash^{40}  & \rewind,\;   & \quad\dots \#1\bolu{\#}20211\#\dots           \\
    \vdash^4\quad & \checkTwo,  &       & \dots \#0\#1\bolu{1}\#\dots\quad    &  &  &  & \vdash^8     & \checkTwo,\; & \quad\dots\# 1\#11100\bolu{1}\#\dots          \\
    \vdash^3\quad & \checkHalt, &       & \dots \# \bolu{0}\#11\#\dots\quad   &  &  &  & \vdash^8     & \rewind,\;   & \quad\dots \#2\bolu{\#}111001\#\dots          \\
                  &             & \quad &                                     &  &  &  & \vdash^{254} & \rewind,\;   & \quad\dots \#2\bolu{\#}1101011202221\#\dots &
\end{align*}
\normalsize

\thmfoursymbols*

\begin{proof}
    \ifllncs
\else
        The machine is called \Mf and is given in Figure~\ref{fig:four}.
We index tape positions by integers in~$\Z$. Initially (at step 0), the tape head is at position $0$ and each position of the tape contains the blank symbol \#. The construction organises the tape as follows: position $-1$ holds a counter to keep track of the 3 known special cases of the Erd\H{o}s' conjecture (1, 4 and 256 in base 10 which are 1, 11 and 100111 in base $3$), position $0$ always contains a blank $\#$ to act as a separator, and positive positions will hold reverse-ternary representations of powers of two, i.e. least significant digit at position $1$ and most significant at position $\lceil \text{log}_3(2^n) \rceil$.  States $\textbf{\mulTwoF}$ and $\textbf{\mulTwoG}$ reproduce the logic of the ``\mulTwo''  FST (Figure~\ref{fig:fst}) where, in state \mulTwoG the blank symbol behave like ternary digit $0$ while in \mulTwoF it triggers the end of the multiplication by $2$.

Starting the Turing machine in state $\textbf{\mulTwoG}$  appropriately initialises the process, the reader can verify that at step $5$ the machine is in state $\textbf{rewind}$, the tape head is at position $0$ and the tape content between positions $-1$ to $2$ included is: $0\underline{\#}1\#$ (tape head position underlined) and all other positions are blank.\footnote{State $\textbf{\mulTwoG}$ was not designed to kick-start the process, but starting with it happens to give us what we need.} From there we prove the following result (IH) by induction on $n\in\N$: at step $s_n=~5~+~c_n~+\sum_{k=1}^{n} 2*(\lceil \text{log}_3(2^k) \rceil+1)$ either (a) \Mf has halted, the tape head is at position $-1$ and the reverse-ternary represented number written on the positive part of the tape (with digits in reverse order) is a counterexample to Erd\H{o}s' conjecture, i.e. it has no digit equal to 2 and is of the form $2^{n_0}$ with $n_0 > 8$, or (b) \Mf is in state \textbf{rewind}, the tape head is at position $0$ which contains a blank symbol and the reverse-ternary represented number by the digits between position $1$ and $\lceil \text{log}_3(2^n) \rceil$ is equal to $2^n$ and all positions coming after $\lceil \text{log}_3(2^n) \rceil$ are blank. The number $c_n$ in the expression of $s_n$ accounts for extra steps that are taken by \Mf to increment the counter at position $-1$ which keeps track of the three known special cases of the conjecture: $1=1_3$, $4 = 11_3$ and $256=100111_3$. In practice, $c_n$ is defined by $c_n = 0$ for $n \leq 1$, $c_n = 2$ for $2 \leq n \leq 7$ and $c_n = 4$ for $n \geq 8$.

For $n = 0$ we have $s_0  = 5$ and case (b) of the induction hypothesis (IH) is verified as $1=2^0$ is represented on the positive part of the tape as seen above.

Let's assume that the result holds for $n\in\N$. If \Mf had already halted at step $s_n$, then case (a) of IH still holds and nothing is left to prove. Othewise, by case (b) of IH, at step $s_n$ then \Mf is in state $\textbf{rewind}$ at position $0$ which contains a \# and between positions $1$ and $\lceil \text{log}_3(2^n) \rceil + 1$ the following word is written: $w = w_1 \dots w_{\lceil \text{log}_3(2^n) \rceil} \#$ such that $w_1\dots w_{\lceil \text{log}_3(2^n) \rceil}$ is the reverse-ternary representation of $2^n$. At step $s_n+1$, tape position $0$ still contains the blank symbol \#, \Mf is in state $\textbf{\mulTwoF}$ and the tape head is at position $1$ where it begins simulating the \mulTwo\ FST on $w$. By having the final $\#$ of $w$ encode a $0$ then Lemma~\ref{lem:fst} applies which means that, after scanning $w$, the positive part of the tape contains the reverse-ternary expression of $2*2^n = 2^{n+1}$. Either \Mf was in state $\textbf{\mulTwoF}$ when it read the final $\#$ of $w$, in which case it will stop simulating the FST and jump to state $\textbf{\checkTwo}$ and move the tape head to the left. Otherwise it was in state $\textbf{\mulTwoG}$ and read the $\#$ in which case it jumps to state $\textbf{\mulTwoF}$ and moves the tape head to the right where, by IH, a $\#$ will be read bringing us back to the previous case. In both cases, $\lceil \text{log}_3(2^{n+1}) \rceil + 1$ symbols have been read since $s_n$, \Mf is in state $\textbf{\checkTwo}$ and all positions after and including $\lceil \text{log}_3(2^{n+1}) \rceil + 1$ are still blank.

State $\textbf{\checkTwo}$ will scan to the left over the reverse-ternary expression that was just computed to search for a digit equal to $2$. If a $2$ is found, then it switches to state $\textbf{rewind}$ and will reach position $0$ at step $s_n + 2*(\lceil \text{log}_3(2^{n+1}) \rceil + 1)$. In that case, we also have $c_{n+1} = c_{n}$ as by definition of $c_n$, one can verify that $c_{n+1} \neq c_{n}$ implies that $2^{n+1}$ has no ternary digit equal to 2. Hence, if a $2$ was found the machine reaches tape position $0$ at step $s_{n+1}=~5~+~c_{n+1}~+\sum_{k=1}^{n+1} 2*(\lceil \text{log}_3(2^k) \rceil+1)$, position $0$ contains the blank symbol, positions $1$ to $\lceil \text{log}_3(2^{n+1}) \rceil$ hold the reverse-ternary representation of $2^{n+1}$ and all positions after $\lceil \text{log}_3(2^{n+1}) \rceil$ are still blank, which is what we wanted under case (b) of IH. If no $2$s were found then \Mf will reach position $-1$ in state $\textbf{\checkHalt}$ and if $n+1 = 2$ or $n+1 = 8$ it will respectively read a $0$ or $1$ which will respectively be incremented to $1$ or $2$, then \Mf goes back to position $0$ in state $\textbf{rewind}$ and case (b) of IH is verified as the value $c_n$ accounts for the extra steps that were taken to increment the counter in those cases. However, if $n+1 > 8$ then \Mf will read a $2$ at position $-1$ and consequently halt with positions $1$ to $\lceil \text{log}_3(2^{n+1}) \rceil$ giving the base three representation of $2^{n+1}$ containing no digit equal to two: a counterexample to Erd\H{o}s' conjecture has been found and it is consequently false and case (a) of IH holds.

From this proof, if \Mf halts then Erd\H{o}s' conjecture is false. For the reverse direction, note that \Mf will test every single successive power of two until it finds a potential counterexample meaning that if Erd\H{o}s' conjecture is false \Mf will find the smallest counterexample and stop there. Hence the 5-state, 4-symbol Turing machine given in Figure~\ref{fig:four} halts if and only if Erd\H{o}s' conjecture is false.
    \fi
\end{proof}


\section{Fifteen states, two symbols Turing machine}
\label{sec:two}

\begin{figure}[!htb]
  \includegraphics[width=\textwidth]{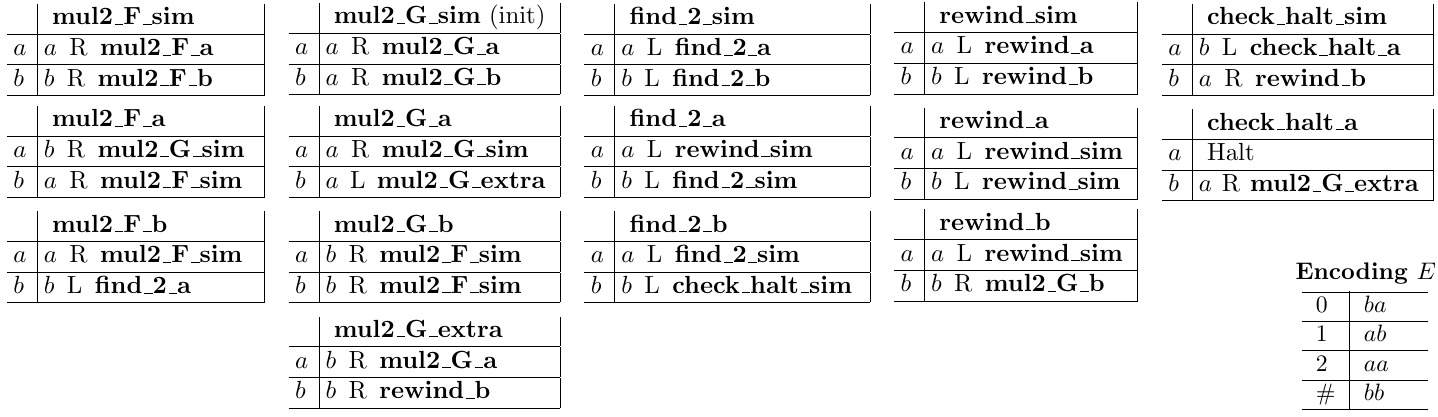}
  \caption{15-state 2-symbol Turing machine  $\Mt$ that halts if and only if Erd\H{o}s' conjecture is false. The initial state is \mulTwoGS{\bf sim},  denoted `(init)'.
    The blank symbol is $b$, and, since this is a busy-beaver candidate, the initial tape is empty:  $\ldots  b b \bolu{b} b b \ldots$  (tape head position bold and underlined).
    States are organised into 5 columns, one for each state of $\Mt$ in Figure~\ref{fig:four}, and inherit name prefixes from $\Mt$.
    In Lemma~\ref{lem:sim} we prove that~$\Mt$ simulates~$\Mf$.}\label{fig:two}
\end{figure}

In this section, we prove Theorem~\ref{thm:two}.
To do this we define the 15-state 2-symbol Turing machine \Mt in Figure~\ref{fig:two} and then prove that it halts if and only if it finds a counterexample to \Eds' conjecture.
We begin with some remarks to guide the reader.

\subsection{Intuition and overview of the construction}
Turing machines with two symbols can be challenging to reason about and prove correctness of for two reasons:
(1) the tape is  a difficult-to-read stream of bits and
(2) simple algorithmic concepts need to be  distributed across multiple states (because of the low number of symbols).
We took the approach of first designing $\Mf$ (Figure~\ref{fig:four}, relatively easy to understand), and then designing
$\Mt$ to simulate $\Mf$. (With the caveat that both machines have a few program-size optimisations.)

In order to avoid confusion,  $\Mt$ uses alphabet $\{a,b\}$ (with  $b$ blank) which is distinct to that of \Mf.  In Lemma~\ref{lem:sim} we prove that \Mt \textit{simulates}  \Mf, via a tight notion of simulation given in Definition~\ref{def:sim} and the Lemma statement; with only a linear time overhead. Symbols of \Mf are encoded by those of \Mt
using the encoding function $E: \{\#,0,1,2\} \to \{a,b\}^2$ defined by $E(\#) = bb$, $E(0) = ba$, $E(1) = ab$, $E(2) = aa$.
Intuitively, the 15 states of $M_{15,2}$ are partitioned into $5$ sets:
the idea is that each of the five states of $\Mf$ is simulated by one of the five corresponding {\em column of states}  in Figure~\ref{fig:two}.

The naming convention for states in $M_{15,2}$ is as follows: for each state of \Mf,  there is a state with the same name in \Mt, followed by {\bf\_sim} (short for `simulate'), that is responsible for initiating the behaviour of the corresponding state in \Mf.
Then, from each  {\bf\_sim} state in $M_{15,2}$ control moves to one of two new states, suffixed by {\bf \_a} and {\bf \_b}, after reading symbol $a$ or symbol $b$.
At the next step, two consecutive letters have been read and $M_{15,2}$ knows which of the $4$ possible cases of the encoding function $E$ it is considering and has sufficient information to simulate one step of $M_{5,4}$.
However, in many cases, for program-size efficiency reasons, \Mt makes a decision before it has read both symbols of the encoding.

There are two exceptions to the \Mt state naming rule: there is no state \checkHaltS {\bf b} as our choice of encoding function $E$ meant not needing to consider that case when simulating \checkHalt, and the state \mulTwoGS{\bf extra} which is involved in simulation of both \mulTwoG and \checkHalt.

\begin{example}
  \label{ex:two}\normalfont
  Here, we highlight 11 of the first 742 configurations of \Mt.
  These configurations  represent the simulation (via encoding $E$) of those highlighted in Example~\ref{ex:four}. For instance, in the second configuration shown we have $bb\;ba\;bb\;ab\;bb\; = E(\#)E(0)E(\#)E(1)E(\#)$ which corresponds to the tape content of the second configuration shown in Example~\ref{ex:four}. Simulation comes with a (merely linear) cost in time, as here, the first 333 steps of \Mf are simulated in 741 steps in \Mt.
  \scriptsize
  \begin{align*}
     &                &  & \mulTwoGS{\bf sim},                 &  & \dots bb\;bb\; \bolu{b}b\;bb\; bb\dots                                                             \\
     & \vdash^{10}    &  & \rewindS{\bf b},                    &  & \dots bb\; ba\; \bolu{b}b\; ab\;bb\dots                                                            \\
     & \vdash^{9\;\;} &  & \rewindS{\bf sim},                  &  & \dots bb\; ab\; b\bolu{b}\; aa\; bb\dots                                                           \\
     & \vdash^{10}    &  & \checkTwoS{\bf sim},                &  & \dots bb\;ba\;bb\;ab\;a\bolu{b}\;bb\dots                                                           \\
     & \vdash^{6\;\;} &  & \checkHaltS{\bf sim}, \hspace{-2ex} &  & \dots bb\; b\bolu{a}\;bb\;ab\;ab\;bb\dots                                                          \\
     & \vdash^3       &  & \rewindS{\bf b},\;                  &  & \dots bb\;ab\;\bolu{b}b\;ab\;ab\;bb\dots                                                           \\
     & \vdash^{13}    &  & \rewindS{\bf sim},\;                &  & \dots bb\;ab\;b\bolu{b}\;aa\;aa\;bb\dots                                                           \\
     & \vdash^{92}    &  & \rewindS{\bf sim},\;                &  & \dots bb\;ab\;b\bolu{b}\;aa\;ba\;aa\;ab\;ab\;bb\dots                                               \\
     & \vdash^{22}    &  & \checkTwoS{\bf sim},\;              &  & \dots bb\; ab\; bb\; ab\;ab\;ab\;ba\;ba\;a\bolu{b}\;bb\dots                                        \\
     & \vdash^{15}    &  & \rewindS{\bf b},\;                  &  & \dots bb\; aa \; \bolu{b}b\; ab\; ab\; ab \; ba \; ba \; ab \; bb\dots                             \\
     & \vdash^{561}   &  & \rewindS{\bf sim},\;                &  & \dots bb\;aa\;b\bolu{b}\; ab\; ab\; ba\; ab\; ba\; ab\; ab\; aa\; ba\; aa\; aa\; aa\; ab\; bb\dots
  \end{align*}

  \normalsize
\end{example}

\ifllncs
\else
  \subsection{Proof of correctness}
  
We define what we mean by `simulates'; the definition couples the dynamics of two machines so that the tape content of one machine is mapped in a straightforward way to that of the other.

\begin{definition}[simulates]\label{def:sim}
  \normalfont
  Let $M$ and $M'$ be single-tape Turing machines with alphabets $\Sigma$ and~$\Sigma'$.
  Then,  $M'$ \textit{simulates}  $M$ if there exists $m\in\N$, a function $E: \Sigma \to \Sigma'^m$ and a computable \textit{time-scaling} function $f: \N \to \N$, such that for all time steps $n\in\N$ of $M$,
  then for step $f(n)$ of $M'$, either
  both $M$ and $M'$ have already halted, or else they are both still running and,
  if $M$ has tape content  $\dots t_{-1}^n t_{0}^n t_1^n \dots$ (where $t_i \in \Sigma$)  then
  $M'$ has tape content $\dots E(t_{-1}^n) E(t_0^n) E(t_1^n) \ldots\;$.
\end{definition}

\begin{definition}[$\Mt$'s encoding function $E$]\label{def:E}
  \normalfont
  Let  $E: \{\#,0,1,2\} \to \{a,b\}^2$ be $\Mt$'s  encoding function, where $E(\#) = bb$, $E(0) = ba$, $E(1) = ab$ and $E(2) = aa$.
\end{definition}

Our main technical lemma in this section states that $\Mt$ (Figure~\ref{fig:two})  simulates $\Mt$ (Figure~\ref{fig:four}), with a linear time overhead:

\begin{lemma}\label{lem:sim}
  \normalfont
  Machine $M_{15,2}$ simulates $M_{5,4}$, according to Definition~\ref{def:sim},
  with encoding function~$E$ (Definition~\ref{def:E}) and time-scaling function $f$ recursively defined by: $f(0) = 0$ and $f(n+1) = f(n) + g((q_n,d_n),\sigma_n)$ with:
  $q_n$ being the state of $M_{5,4}$  at step $n$, $d_n \in \{\text{L},\text{R}\}$ the tape head direction at step $n-1$ (where $d_0 = \text{R}$), $\sigma_n \in  \{\#,0,1,2\}$ the read symbol at step $n$, and $g$ the partial function defined in Figure~\ref{fig:timescale}.
\end{lemma}

\begin{proof}

  Let $S=\{\mulTwoF, \mulTwoG, \checkTwo, \rewind, \checkHalt\}$ be the set of 5 states of $M_{5,4}$ and $S'$ be the set of 15 states of $\Mt$ given in Figure~\ref{fig:two}. We prove the result by induction on $n\in\N$, the number of steps taken by the {\em simulated} machine $\Mf$. Let $k_n = (q_n,d_n)\in S \times \{\text{L},\text{R}\}$ be $\Mf$'s state at step~$n$ together with the tape-head direction at the {\em previous} step $n-1$ (we take the convention $d_0 = \text{R}$), let $\sigma_n \in \{0,1,2,\#\}$ be $\Mf$'s read symbol at step $n$, and let $i_n\in\Z$ be $\Mf$'s tape head position at step $n$.
  We note that by inspecting Figure~\ref{fig:four} the set of possible values for $k_n$ is
  \begin{equation*}
    K=\{(\mulTwoF,\text{R}),(\mulTwoG,\text{R}),(\checkTwo,\text{L}),(\checkHalt,\text{L}),(\rewind,\text{L}),(\rewind,\text{R})\}
  \end{equation*}
  observing that \rewind is the only $M_{5,4}$ state reachable  from both left and right tape head moves.

  The induction hypothesis (IH) is the definition of simulation (Definition~\ref{def:sim}) instantiated with~$E$ (Definition~\ref{def:E}),  $f$ as defined in the lemma statement, and the following: at step $f(n)$ machine $M_{15,2}$'s head is at position $2i_n + \pi(d_n)$ with $\pi(L) = 1$ and $\pi(R) = 0$ and the machine is in state $h(k_n)$ with $h: K \to S'$ defined by $h(\mulTwoF,\text{R})=\mulTwoFS{\bf sim}$, $h(\mulTwoG,\text{R})=\mulTwoGS{\bf sim}$, $h(\checkTwo,\text{L})=\checkTwoS{\bf sim}$, $h(\checkHalt,\text{L}) = \checkHaltS{\bf sim}$, $h(\rewind,\text{L})=\rewindS{\bf sim}$ and $h(\rewind,\text{R})=\rewindS{\bf b}$.

  If $n=0$, we have $f(0) = 0$ and both tapes are entirely blank: $\dots \# \# \# \dots$ for \Mf, and $\dots bbb \dots$ for \Mt which is consistent with $E(\#) = bb$. Additionally, $k_0 = (\mulTwoG,R)$, the tape head of \Mt is at position $0 = 2i_0 + 0$ with $i_0 = 0$ being the tape head position of \Mf and finally, \Mt is in state $h(k_0) = \mulTwoGS{\bf sim}$ (which is also its initial state).

  Let's assume that the induction hypothesis (IH) holds for $n \in \N$. If at step $n$ for \Mf, and step $f(n)$ for \Mt, both machines have already halted they will still have halted at any future step and the IH still holds.
  Let's instead assume from now on, IH holds and we are in the second case of Definition~\ref{def:sim} (both are still running and tape contents are preserved under the encoding function $E$).

  If $k_n = (\mulTwoF,\text{R})$ and $\sigma_n = 2$ then the configuration of \Mf at step $n$ is $ \mulTwoF,\; \allowbreak \dots \bolu{2} \dots$ with tape head at position $i_n$. By Figure~\ref{fig:four}, at step $n+1$ the configuration becomes $\mulTwoG,\; \dots 1\bolu{*} \dots$ and $i_{n+1} = i_n + 1$  (with $*$ being whatever symbol is at position $i_n + 1$) and $k_{n+1} = (\mulTwoG,\text{R})$. By IH, at step $f(n)$ the tape head position of \Mt is $2i_n + \pi(\text{R}) = 2i_n$ and the configuration of \Mt is $ \mulTwoFS{\bf sim},\; \dots \bolu{a}a \dots$, since $E(2) = aa$ and $h(k_n) = \mulTwoFS{\bf sim}$. By Figure~\ref{fig:two}, at step $f(n+1) = f(n) + g((\mulTwoF,\text{R}),1) = f(n) + 2$, the configuration of \Mt is $ \mulTwoGS{\bf sim},\; \dots ab\bolu{*} \dots$ with tape head  at position $2i_n +2$. We have $E(1) = ab$ (and no other tape positions than $2i$ and $2i+1$ were modified), $2i_n + 2 = 2(i_n + 1) = 2i_{n+1} + \pi(\text{R}) = 2i_{n+1} + \pi(d_{n+1})$ and $\mulTwoGS{\bf sim} = h(k_{n+1})$ which is everything we needed to satisfy IH at step $n+1$.
  We leave verifications of cases $\sigma_n = 0$ and $\sigma_n = 1$ to the reader as they are very similar. If $\sigma_n = \#$ then \Mf writes $\#$ then goes left to state $\checkTwo$, which gives $k_{n+1} = (\checkTwo,\text{L})$ and $i_{n+1} = i_n - 1$. By IH, at step $f(n)$ the tape head position of \Mt is $2i_n + \pi(\text{R}) = 2i_n$ and the configuration of \Mt is $ \mulTwoFS{\bf sim},\; \dots \bolu{b}b \dots$. After $g((\mulTwoF,\text{R}),\#)=3$ steps, the configuration becomes $ \checkTwoS{\bf sim},\; \dots \bolu{*}bb \dots$. This is consistent with having $E(\#) = bb$, moving the tape head to position $2i_n - 1 = 2(i_n-1)+1 = 2i_{n+1} + \pi(d_{n+1})$ and having $\checkTwo{\bf sim} = h(k_{n+1})$ which is all we needed.

  \begin{figure}[t] 
    \centering\scriptsize
    \setlength\tabcolsep{4.5pt}
    \begin{tabular}{c|c|c|c|c|c|c|}
         & (\mulTwoF,R) & (\mulTwoG,R)\hspace{.3ex} & (\checkTwo,L) & (\checkHalt,L) & (\rewind,L) & (\rewind,R) \\ \hline
      0  & 2            & 2                         & 2             & 3              & 2           & -           \\ \hline
      1  & 2            & 4                         & 2             & 1              & 2           & -           \\ \hline
      2  & 2            & 2                         & 2             & 2              & 2           & -           \\ \hline
      \# & 3            & 2                         & 2             & 1              & 3           & 2           \\ \hline
    \end{tabular}
    \caption{The partial function $g$ of Lemma~\ref{lem:sim} that defines how many steps are needed by $\Mt$ to simulate one step of \Mf, for each (state, move-direction-on-previous-step), and symbol, of an \Mf step.
      The function is partial, defined on only one entry of the final column, as the proof of Theorem~\ref{thm:four} shows that only symbol $\#$ can be read if \Mf reaches state \rewind~coming from the left (hence by moving tape head in direction R).
    }\label{fig:timescale}
  \end{figure}

  If $k_n = (\mulTwoG,\text{R})$ and $\sigma_n = 1$ then the configuration of \Mf at step $n$ is
  $ \mulTwoG,\; \dots \bolu{1} \dots$ with tape head at position $i_n$. By Figure~\ref{fig:four}, at step $n+1$ the configuration becomes $ \mulTwoG,\; \dots 0\bolu{*} \dots$ and $i_{n+1} = i_n + 1$ and $k_{n+1} = (\mulTwoG,\text{R})$. By IH, at step $f(n)$ the tape head position of \Mt is $2i_n + \pi(\text{R}) = 2i_n$ and the configuration of \Mt is $ \mulTwoGS{\bf sim},\; \dots \bolu{a}b \dots$ as $E(1) = ab$ and $h(k_n) = \mulTwoGS{\bf sim}$. In that case, the machine will need $g((\mulTwoG,\text{R}),1) = 4$ steps as it will have to bo back one step to the left after having scanned the second symbol of $E(1)=ab$ in order to write the first symbol $E(0) = ba$. This is realised by the first transition of intermediate state \mulTwoGS{\bf extra} as we are assured to read symbol $a$ in that case. Altogether, we get that at step $f(n+1) = f(n) + 4$ the configuration of \Mt is $ \mulTwoGS{\bf sim},\; \dots ba\bolu{*} \dots$ with $\mulTwoGS{\bf sim}=h(k_{n+1})$ and tape head at position $2i_n +2 = 2i_{n+1} + \pi(d_{n+1})$ which is everything we need. We leave cases $\sigma_n \in \{0,2,\#\}$ to the reader.

  If $k_n = (\checkTwo,\text{L})$, the simulation is straightforward for any $\sigma_n \in \{0,1,2,\#\}$ as the content of the tape is not modified and the tape head always moves to the left, hence we leave those cases to the reader.

  If $k_n = (\rewind,\text{L})$ and $\sigma_n = \#$ then the configuration of \Mf at step $n$ is
  $ \mulTwoG,\; \dots \bolu{\#} \dots$ with tape head at position $i_n$. By Figure~\ref{fig:four}, at step $n+1$ the configuration becomes $\mulTwoF,\; \dots \#\bolu{*} \dots$ and $i_{n+1} = i_n + 1$ and $k_{n+1} = (\mulTwoF,\text{R})$. By IH, at step $f(n)$ the tape head position of \Mt is $2i_n + \pi(\text{L}) = 2i_n + 1$ and the configuration of \Mt is $ \rewindS{\bf sim},\; \dots b\bolu{b} \dots$ as $E(\#) = bb$ and $h(k_n) = \rewindS{\bf sim}$. By Figure~\ref{fig:two}, at step $f(n+1) = f(n) + g((\rewind,\text{L}),\#) = f(n) + 3$, the configuration of \Mt is $\mulTwoFS{\bf sim},\; \dots bb\bolu{*} \dots$ with $bb=E(\#)$, $\mulTwoFS{\bf sim}=h(k_{n+1})$ and  tape head  at position $2i_n +2 = 2i_{n+1}+\pi(d_{n+1})$, which is everything that we need. We leave cases $\sigma_n \in \{0,1,2\}$ to the reader as in those cases the tape head always moves to the left and the tape content is also not modified.

  If $k_n = (\rewind,\text{R})$, by the proof of Theorem~\ref{thm:two}, we know that necessarily $\sigma_n = \#$, since in that case, the tape head of \Mf is at position $i_n = 0$ which always holds a $\#$ (the machine is just after incrementing the `3 special cases of \Eds' conjecture' counter at position $-1$), hence the configuration of \Mf at step $n$ is $ \rewind,\; \dots \bolu{\#} \dots$ and at step $n+1$ it becomes $ \mulTwoF,\; \dots \#\bolu{*} \dots$ (Figure~\ref{fig:four}) and $i_{n+1} = i_n + 1$ and $k_{n+1} = (\mulTwoF,\text{R})$. By IH, at step $f(n)$ the tape head position of \Mt is $2i_n + \pi(\text{R}) = 2i_n$ and the configuration of \Mt is $ \rewindS{\bf b},\; \dots \bolu{b}b \dots$ as $E(\#) = bb$ and $h(k_n) = \rewindS{\bf b}$. By Figure~\ref{fig:two}, at step $f(n+1) = f(n) + g((\rewind,\text{R}),\#) = f(n) + 2$, the configuration of \Mt is $ \mulTwoFS{\bf sim},\; \dots bb\bolu{*} \dots$ with $bb=E(\#)$, $\mulTwoFS{\bf sim}=h(k_{n+1})$ and the tape head is at position $2i_n +2 = 2i_{n+1}+\pi(d_{n+1})$, which is everything that we need.

  If $k_n = (\checkHalt,\text{L})$ then note that for $\sigma_n \in \{0,1,\#\}$ machine \Mf will `increment' the value of $\sigma_n$ then transition to $k_{n+1} = (\rewind, \text{R})$ moving its head to $i_{n+1} = i_n + 1$. Cases $\sigma_n \in \{\#,1\}$ are arguably where we make the most use of the encoding function $E$ in order to use very few states in \Mt to simulate the behavior of \Mf. Indeed, if $\sigma_n = \#$ or $\sigma_n = 1$ then \Mf must respectively write $1$ and $2$. This means that encodings must go from $E(\#) = bb$ to $E(1) = ba$ and from $E(1) = ab$ to $E(2) = aa$. By IH, machine \Mt is currently reading the second symbol of the encoding (because the head is at position $2i_n + \pi(\text{L}) = 2i_n + 1$), which is symbol $b$, then it is enough for \Mt, using only one step, just to turn that $b$ into an $a$ which deals with both cases and then go right to state $h(k_{n+1}) = \rewindS{\bf b}$ while moving the head to $2i_n + 2 = 2i_{n+1} + \pi(d_{n+1})$ which is what we need. This gives $g((\checkHalt,\text{L}),\#) = g((\checkHalt,\text{L}),1) = 1$. We let the reader verify the case $\sigma_n = 0$ which gives $g((\checkHalt,\text{L}),0) = 3$.

  If $k_n = (\checkHalt,\text{L})$ and $\sigma_n = 2$ then, at step $n+1$ machine \Mf will halt. By IH, at step $f(n)$ the tape head position of \Mt is $2i_n + \pi(\text{L}) = 2i_n + 1$ and the configuration of \Mt is $ \checkHaltS{\bf sim},\; \dots a\bolu{a} \dots$ as $E(2) = aa$ and $h(k_n) = \checkHaltS{\bf sim}$. By Figure~\ref{fig:two}, at step $f(n+1) = f(n) + g((\checkHalt,\text{L}),2) = f(n) + 2$ machine \Mt halts. Hence if \Mf halts then \Mt halts. Machine \Mt has only one halting instruction and this proof shows that it is reached only in the case where $k_n = (\checkHalt,\text{L})$ and $\sigma_n = 2$ meaning that if \Mt halts then \Mf halts. Hence the machines are either both running at respectively step $n+1$ and $f(n+1)$ or both have halted.

  In all cases, the induction hypothesis is propagated at step $n+1$ and we get the result: machine \Mt simulates \Mf according to Definition~\ref{def:sim}.
\end{proof}
\fi
\subsection{Main result and corollaries}\label{sec:main and cor}
\ifllncs
\else
  Using our previous results, the proofs of Theorem~\ref{thm:two}, and Corollary~\ref{cor:ErdosIsFinite} and~\ref{cor:BBlowerBound} are almost immediate:
\fi
\thmtwosymbols*
\begin{proof}
  By Lemma~\ref{lem:sim},  \Mt simulates \Mf which in turns means that \Mt halts if and only if \Mf halts.
  By Theorem~\ref{thm:four} this means that  \Mt halts if and only if Erd\H{o}s' conjecture is false.
\end{proof}

\ErdosIsFinite*
\begin{proof}
  From the proof of Theorem~\ref{thm:four},
  If we run \Mf
  then, at step $\BB(5,4)+1 \in \mathbb{N}$ we know if Erd\H{o}s' conjecture is true or not:
  it is true if and only if \Mf is still running.
  If \Mf is still running, and because the machine outputs fewer than one power of $2$ per step, at step $\BB(5,4)+1$ the power of $2$ written on the tape is at most $2^{\BB(5,4)}$.
  Analogously, \Mt writes fewer than one  power of $2$ per step, hence, at step $\BB(\beatit)+1\in \mathbb{N}$
  the power of $2$ written on the tape is at most $2^{\BB(\beatit)}$.
  In either case, by then, we know whether the conjecture is true or not.
  Hence, it is enough to check Erd\H{o}s' conjecture for all $n \leq \min( \BB(15), \BB(5,4))$ and we get the result.
\end{proof}

\BBlowerBound*
\begin{proof}
  By Theorems~\ref{thm:two} and~\ref{thm:four}, if \Eds' conjecture has a counterexample $x$, then we have explicit busy beaver contenders for $\BB(\beatit)$ and $\BB(5,4)$, respectively, i.e.~machines $\Mt$ and $\Mf$. These
  machines both halt with $x = 2^{n_0}$,   $n_0\in\mathbb{N}$, written on their tape in base 3 (by the proofs of Theorems~\ref{thm:two} and~\ref{thm:four}).
  $\Mt$ and $\Mf$ enumerate less than one power of 2 per time step, hence their running time is $\geq n_0 = \log_2 x$, giving the stated lowerbound on  $\BB(\beatit)$ and $\BB(5,4)$.
\end{proof}

\vspace{-4ex}
\ifllncs
\else
  \bibliographystyle{abbrv}
\fi
\bibliography{main}

\begin{thebibliography}{10}

\bibitem{bbchallengeBusyBeaver}
{T}he {B}usy {B}eaver {C}hallenge.
\newblock \href{https://bbchallenge.org/}{https://bbchallenge.org/}.
\newblock [Accessed 27-05-2024].

\bibitem{BusyBeaverFrontier}
S.~Aaronson.
\newblock The busy beaver frontier.
\newblock {\em SIGACT News}, 51(3):32–54, Sept. 2020.
\newblock Preprint: \url{https://www.scottaaronson.com/papers/bb.pdf}.

\bibitem{baseChanges}
B.~Adamczewski and C.~Faverjon.
\newblock Mahler's method in several variables {II}: Applications to base change problems and finite automata, Sept. 2018.
\newblock Preprint: \url{https://arxiv.org/abs/1809.04826}.

\bibitem{10.5555/57249.57258}
A.~H. Brady.
\newblock The busy beaver game and the meaning of life.
\newblock In {\em A Half-Century Survey on The Universal Turing Machine}, page 259–277, USA, 1988. Oxford University Press, Inc.

\bibitem{chaitin1987computing}
G.~J. Chaitin.
\newblock Computing the busy beaver function.
\newblock In {\em Open Problems in Communication and Computation}, pages 108--112. Springer, 1987.

\bibitem{Charles2013}
{Code Golf Addict}.
\newblock list27.txt.
\newblock \newline\url{https://gist.github.com/anonymous/a64213f391339236c2fe31f8749a0df6}, 2016.

\bibitem{mawatam}
M.~Cook, T.~St\'erin, and D.~Woods.
\newblock Small tile sets that compute while solving mazes.
\newblock In M.~Lakin and P.~Sulc, editors, {\em Proceedings of the 27th International Conference on DNA Computing and Molecular Programming (DNA 27)}, volume 205 of {\em Leibniz International Proceedings in Informatics (LIPIcs)}, pages 1--20, Dagstuhl, Germany, 2021. Schloss Dagstuhl--Leibniz-Zentrum f{\"u}r Informatik.
\newblock Arxiv preprint: \url{https://arxiv.org/abs/2106.12341}.

\bibitem{dimitrov2021powers}
V.~S. Dimitrov and E.~W. Howe.
\newblock Powers of 3 with few nonzero bits and a conjecture of {Erd\H{o}s}, 2021.
\newblock \url{https://arxiv.org/abs/2105.06440}.

\bibitem{DUPUY2016268}
T.~Dupuy and D.~E. Weirich.
\newblock Bits of $3^n$ in binary, {W}ieferich primes and a conjecture of {Erd\H{o}s}.
\newblock {\em Journal of Number Theory}, 158:268--280, 2016.

\bibitem{ErdosPowers2}
P.~Erd\H{o}s.
\newblock Some unconventional problems in number theory.
\newblock {\em Mathematics Magazine}, 52(2):67--70, 1979.

\bibitem{Skelet_website}
G.~Georgiev.
\newblock {Busy Beaver prover}.
\newblock \url{https://skelet.ludost.net/bb/index.html} [Accessed 27-05-2024].

\bibitem{HARLAND2022368}
J.~Harland.
\newblock Generating candidate busy beaver machines (or how to build the zany zoo).
\newblock {\em Theoretical Computer Science}, 922:368--394, 2022.

\bibitem{BB748Thesis}
{Johannes Riebel}.
\newblock {The Undecidability of BB(748)}.
\newblock Bachelor's thesis, 2023.
\newblock \url{https://www.ingo-blechschmidt.eu/assets/bachelor-thesis-undecidability-bb748.pdf}.

\bibitem{Pavel_discorvery}
P.~Kropitz.
\newblock {BB(6, 2) $>$ 4\^{}4\^{}4\^{}4\^{}7}, 2022.
\newblock \url{https://groups.google.com/g/busy-beaver-discuss/c/-zjeW6y8ER4/m/ZBuLvbVOAgAJ}.

\bibitem{Pavel_discorvery2}
P.~Kropitz.
\newblock {BB(2, 6) $>$ 10$\uparrow\uparrow$10$\uparrow\uparrow$10$\uparrow\uparrow$3}, 2023.
\newblock \url{https://groups.google.com/g/busy-beaver-discuss/c/UuC\_Yjc5LPQ}.

\bibitem{Pavel_discorvery3}
P.~Kropitz.
\newblock {BB(3, 4) $>$ Ack(14)}, 2024.
\newblock \url{https://groups.google.com/g/busy-beaver-discuss/c/dkecbR5b5Og/}.

\bibitem{lagariasErdos}
J.~C. Lagarias.
\newblock Ternary expansions of powers of 2.
\newblock {\em J. Lond. Math. Soc. (2)}, 79(3):562--588, 2009.

\bibitem{Ack14BlogPost}
S.~Ligocki.
\newblock {BB(3, 4) > Ack(14)}.
\newblock \href{https://www.sligocki.com/2024/05/22/bb-3-4-a14.html}{{https://www.sligocki.com/2024/05/22/bb-3-4-a14.html}}, 05 2024.

\bibitem{Marxen_1998}
H.~Marxen and J.~Buntrock.
\newblock {Attacking the Busy Beaver 5}.
\newblock {\em Bull. EATCS}, 40:247--251, 1990.

\bibitem{PMichel_website}
P.~Michel.
\newblock {The Busy Beaver Competitions}.
\newblock \ \\ \href{https://bbchallenge.org/~pascal.michel/bbc.html}{https://bbchallenge.org/$\sim$pascal.michel/bbc.html} [Accessed 27-05-2024].

\bibitem{michel2014simulation}
P.~Michel.
\newblock {Simulation of the Collatz 3x+1 function by Turing machines}.
\newblock Technical report, 2014.
\newblock Arxiv preprint: \href{https://arxiv.org/abs/1409.7322}{https://arxiv.org/abs/1409.7322}.

\bibitem{michel2019busy}
P.~Michel.
\newblock {The Busy Beaver Competition: a historical survey}.
\newblock Technical report, Sept. 2019.

\bibitem{neary2009four}
T.~Neary and D.~Woods.
\newblock Four small universal {Turing} machines.
\newblock {\em Fundamenta Informaticae}, 91(1):123--144, 2009.

\bibitem{Rado_1962}
T.~Rad\'{o}.
\newblock On non-computable functions.
\newblock {\em Bell System Technical Journal}, 41(3):877–884, 1962.
\newblock \href{https://archive.org/details/bstj41-3-877/mode/2up}{https://archive.org/details/bstj41-3-877/mode/2up}.

\bibitem{rogozhin1996small}
Y.~Rogozhin.
\newblock Small universal {T}uring machines.
\newblock {\em Theoretical Computer Science}, 168(2):215--240, 1996.

\bibitem{thesis_sixtiles}
T.~St\'{e}rin.
\newblock {\em {Six Tiles: from Collatz Sequences to Algorithmic DNA Origami}}.
\newblock PhD thesis, Maynooth University, 2023.

\bibitem{Collatz2}
T.~St{\'e}rin and D.~Woods.
\newblock {The {C}ollatz process embeds a base conversion algorithm}.
\newblock In S.~Schmitz and I.~Potapov, editors, {\em Reachability Problems}, pages 131--147. Springer International Publishing, 2020.

\bibitem{TerenceTaoBlog}
T.~Tao.
\newblock {The Collatz conjecture, Littlewood-Offord theory, and powers of~2 and~3}.
\newblock Last accessed June 30th 2021. \href{https://terrytao.wordpress.com/2011/08/25/the-collatz-conjecture-littlewood-offord-theory-and-powers-of-2-and-3/}{https://terrytao.wordpress.com/2011/08/25/the-collatz-conjecture-littlewood-offord-theory-and-powers-of-2-and-3/}.

\bibitem{WoodsNeary2006B}
D.~Woods and T.~Neary.
\newblock On the time complexity of 2-tag systems and small universal {Turing} machines.
\newblock In {\em In 47$^{\textrm{th}}$ Annual IEEE Symposium on Foundations of Computer Science (FOCS)}, pages 132--143, Berkeley, California, Oct. 2006. IEEE.

\bibitem{WoodsNearySurvey}
D.~Woods and T.~Neary.
\newblock The complexity of small universal {T}uring machines: A survey.
\newblock {\em Theoretical Computer Science}, 410(4-5):443--450, 2009.

\bibitem{Yedidia2016}
A.~Yedidia and S.~Aaronson.
\newblock A relatively small {T}uring machine whose behavior is independent of set theory.
\newblock {\em Complex Systems}, 25(4):297--328, Dec. 2016.

\end{thebibliography}

\end{document}